\documentclass[letterpaper, 10 pt, journal]{ieeetran}  

\usepackage{graphics} 
\usepackage{epsfig} 
\usepackage{mathptmx} 
\usepackage{times} 
\usepackage{amsmath} 
\usepackage{amssymb}  
\usepackage{graphicx} 
\usepackage{stackrel}  
\usepackage{color}    %
\usepackage{epstopdf} 
\usepackage{cite}
\usepackage{mathrsfs}
\usepackage{latexsym}
\usepackage{amsthm}
\usepackage{amsfonts}
\usepackage{subfig}
\usepackage{lipsum}
\usepackage{cases}   
\usepackage{float}
\usepackage{arydshln}
\usepackage[pagebackref=false,colorlinks,linkcolor=blue,citecolor=magenta]{hyperref}  

\title{\LARGE \bf
  Bearing and Distance Formation Control of Rigid Bodies in $SE(3)$ with Bearing and Distance Constraints
}

\author{Sara Mansourinasab$^{1}$, Mahdi Sojoodi$^{2}$ and S. Reza Moghadasi$^{3}$
	\thanks{$^{1}$Sara Mansourinasab is with the Department of Electrical and Computer Engineering, Tarbiat Modares University, Tehran, Iran 
		{\tt\small sara.mansouri@modares.ac.ir}}%
	\thanks{$^{2}$Mahdi Sojoodi is with the Department of Electrical and Computer Engineering, Tarbiat Modares University, Tehran, Iran
		{\tt\small sojoodi@modares.ac.ir}}%
	\thanks{$^{3}$S. Reza Moghadasi is with the Department of Mathematical Sciences, Sharif University of Technology, Tehran, Iran
		{\tt\small moghadasi@sharif.ir}}%
}

\begin{document}

\maketitle
\thispagestyle{empty}
\pagestyle{empty}

\begin{abstract}
	Rigidity of the interaction graph is a fundamental condition for achieving the desired formation which can be defined in terms of distance or bearing constraints between agents. In this paper, for reaching a unique formation with the same scaling and orientation as the target formation, both distance and bearing constraints are considered for defining the desired formation. Besides, both distance and bearing measurements are also available. Each agent is able to gather the measurements with respect to other agents in its own body frame. So, the agents are coordinated-free concerning a global reference frame. On the other hand, the framework is embedded in $SE(3)$. The control signal is designed based on a gradient descent method by introducing a cost function. Firstly, the formation problem is considered for bearing-only constraints in $SE(3)$ configuration. Then, the formation control is expressed for the general case of both bearing and distance constraints. Furthermore, the essential conditions that guarantee reaching the desired formation is discussed. Finally, the validity of the proposed formation control is verified by numerical simulations.
\end{abstract}
\section{Introduction}
Formation control is an actively studied strategy in analyzing multi agents systems \cite{oh2015survey,liu2020robust}. In formation control, the type of available data that agents have access plays an important role in designing the control strategy. Based on sensing capabilities for measuring the relative positions between agents, formation control methods are categorized into three types of position-based, displacement-based, and distance-based. The two former methods are based on global distance measuring which necessitates the use of global positioning system (GPS), for example. Then, each agent shares the information with other agents through wireless communications \cite{dorfler2009formation}. Using GPS is not always reliable since limitations such as high dependence of measurements accuracy on the number of available satellites, environmental challenges like indoors, deep urban canyons, dense vegetation and cloud cover obstruct line-of-sight to the satellite \cite{ragothaman2021urban}.
 In distance-based method, the distances between agents are controlled to achieve the desired formation. In this method, the sensing capability is only defined with respect to local coordinate systems or body frame of agents and no GPS is required. Local coordinates formation control procedures include distance-based and bearing-based methods depending on the type of measuring and constraint parameters. 
 Formation rigidity theory helps to achieve the desired formation up to scaling, translation and coordinated rotation factors by its inter-agent bearings and distances. According to the definitions, in a given system connected by flexible linkages and hinges, rigidity has addressed the amount of stiffness of the given framework to an induced deformation \cite{michieletto2016bearing}. In a rigid formation, converging to the desired formation can be achieved by distance and/or bearing measurements. Besides, the desired formation can be defined based on distance and/or bearing constraints between any pair of agents. These types of controllers are classified as distance/bearing-based controller. In distance-based formation  control, each agent has its own body frame which does not need to be aligned with other agents body frame orientation.

Distance-based rigidity problem has been studied in many investigations \cite{anderson2008rigid,kyo2011formation,liang2016distance}. In most of them, the target formation constraints are only defined as inter-agent distances. In bearing-based theory, the desired formation constraints are defined in terms of the direction of the neighboring agents. In other words, the direction of the unit vector aligned the edge that connecting two agents is considered as inter-agent bearing between those two agents. \cite{sun2021vision,zhao2019bearing} have extensively studied different multi-agent strategies based on bearing rigidity theory for converging to the desired formation, generally assume each agent is able to measure bearings and distances of the nearest neighbors. Approaches with only  bearing measurements capability are proposed as bearing-only formation control problems. Vision-based devices are suitable for bearing measurements. Optical cameras are one of low-cost and light-weight onboard sensors that are bearing-only sensors. This yields the resulting formation to be with a different scale of the desired formation. 

In this paper, formation control for a multi-rigid body system is proposed considering bearing measurements and constraints. Since the motion behavior of a rigid body is a combination of translation and rotation movements, the notion of bearing rigidity comes into play for synchronized target formations. As an illustrative example, multi-satellite systems that are orbiting around the Earth, should have a specific orientation to cover a specific region on Earth. As a result, in multi-rigid body problems, using bearing rigidity concept is beneficial, more accurate and low-cost. Geometrically speaking, roto-translational motion of a rigid body represents an $SE(3)$ state space configuration which is a Lie group. Using bearing measurements for analyzing such systems is compatible with geometry of the problem \cite{zelazo2014rigidity}.

In bearing rigidity problems, the  only bearing preserving motions subtend translation and scaling of the entire system \cite{zhao2016bearingalmost}. However, distance preserving motions in distance-based problems include roto-translational motions. It is clear that the rotation of the resulting formation in bearing-based constraints problem will change the inter-agent bearings. On the other hand, the configuration space of the problem is embedded in Special Euclidean space $SE(3)$. Considering the geometry and the nature of the $SE(3)$ manifold enables the rotation of the entire framework. This result in a synchronous rotation of all agents around their body attached framework with the same angular velocity that is coupled with the entire graph\cite{michieletto2016bearing}.\\
\indent It can be concluded that in analysing bearing rigidity problem on $\mathbb{R}^3$, the only admissible infinitesimally rigid motions include translation and scaling of the framework. However, for analysing the problem on $ SE(3)$, the only admissible infinitesimally rigid motions include translation, scaling and coordinate rotation. besides, in distance preserving problems, the only admissible infinitesimally rigid motions include translation and coordinate rotation of the framework. Due to this reason, in bearing-only formation problem on $ \mathbb{R}^3$, the graph is able to converge to the target formation but with a different scale. Furthermore, in distance-based problems, the final graph may be a rotated form of the target graph. Since in some problems which the final formation should be exactly the same as the desired formation, the combination of bearing and distance rigidity is applied. As a result, in this paper, the formation control problem on $SE(3)$ is considered in two cases of bearing constraints and bearing-distance constraints in order that the framework uniquely reach to the desired formation.\\
\indent Since implementing bearing measurements in local frame is compatible with more realistic and applicable scenarios, \cite{zelazo2015bearing} have addressed the position and orientation of a graph of rigid bodies with $SE(2)$ architecture, which restricts the results to the plane. Another local body frame scheme is found in \cite{liu2020robust} with 3 dimensional extension. However, the agents attitude rotation is limited to just one axis. The aforementioned problem is eliminated in \cite{michieletto2016bearing} by considering fully-actuated multi-agents formations with six controllable degrees of freedom.\\
\indent The bearing rigidity problem has been mostly studied in 2-dimensional spaces. The extension of this problem to 3D is formulated in \cite{zhao2016bearingalmost}. The relation and comparison between bearing and distance rigidity is described in detail. This paper defines the bearing based formation problem by bearing-only measurements and provides the bearing-only control in order to reach formation. Besides, it demonstrates under what conditions a framework can be achieved by bearing measurements with translation and scaling factors. However, this paper is confined to the Euclidean $\mathbb{R}^3$ space. \cite{michieletto2021unified} is a comprehensive reference aims at expressing bearing rigidity theory concept of a framework in $ \mathbb{R}^n$ spaces. It redefines the concept of bearing rigidity to the more complicated manifold spaces than n-dimensional Euclidean space. Although the paper demonstrates the main notions of bearing rigidity theory in $SO(n)$ and $SE(n)$ which has been cited by the current paper, it is deprived of any controlling interpretations to steer the framework to the target formation. Therefore, in the presented paper, referring to the results of \cite{michieletto2021unified}, we present the required control strategy to convey the agents  to the desired formation in $SE(3)$ space. This strategy is investigated on two cases of bearing-only and mixed bearing-distance desired constraints.

The paper is organized as follows. Section II is dedicated to review the required elements of graph and bearing rigidity theory. The initial problem formulation for rigid body agents  is defined in the second part of this section. In section III, the desired formation is considered as bearing-only constraints and the control law is designed based on relative bearing measurements in local body frame to achieve the desired formation. In section IV, the formation is introduced as bearing-distance measurements in order to provide the exact shape as the desired one with no scale and rotation. In section V, the results are supported with computer simulations.
\section{Preliminaries and Problem Formulation}
\subsection{Graph Theory}
In this section, some elements of graph theory and bearing rigidity theory are briefly reviewed. A graph of $n$ agents with $m$ edges is defined as the pair $ \mathcal{G}=(\mathcal{V}, \mathcal{E})$ while $\mathcal{V}=\{v_1, ..., v_n\}$ is the vertex set and $ \mathcal{E}=\{ e_1,...,e_m \} \subseteq \mathcal{V} \times \mathcal{V}$ is the edge set. An edge $e_k=(v_i,v_j) \in \mathcal{E} $, while an information exchange such as distance measurements or bearing measurements exists between two vertices $ (i,j)$, and $k$ relates to the $k$-th directed edge in the graph. Each vertex $ v_i \in \mathcal{V}$ in the graph is associated to the point $p_i \in \mathbb{R}^3$ in the configuration, so the framework $(\mathcal{G},p) $ in $\mathbb{R}^3$ consists of the graph $ \mathcal{G}=(\mathcal{V}, \mathcal{E})$ and the configuration $p=[p_1^T ... p_n^T]^T \in \mathbb{R}^{3n}$. The neighborhood of the vertex $i$ is denoted by $ \mathcal{N}_i= \{ j \in \mathcal{V} | (i,j) \in \mathcal{E}\}$.
In a directed framework $\mathcal{G}=(\mathcal{V}, \mathcal{E})$, the incidence matrix $E \in \mathbb{R} ^{n \times m}$ is defined as
\begin{align}
[E]_{ik} = \left\{ \begin{array}{rl}
-1 &\text{if} \ e_k=(v_i,v_j) \in \mathcal{E} \ (\text{outgoing edge}) \\
1 &\text{if} \ e_k=(v_j,v_i) \in \mathcal{E} \ (\text{ingoing edge}) \\
0 &\text{otherwise} \ \ \ \ 
\end{array} \right.  \label{eq5}
\end{align}
and the matrix $ E_o \in \mathbb{R}^{n \times m} $ is introduced by
\begin{align}
[E_o]_{ik} = \left\{ \begin{array}{rl}
-1 &\text{if} \ e_k=(v_i,v_j) \in \mathcal{E} \ (\text{outgoing edge}) \\
0 &\text{otherwise}\ \ \ \ 
\end{array} \right.  \label{eq6}
\end{align}
while $ \bar{E}=E \otimes I_3 \in \mathbb{R} ^{3n \times 3m } $, $ \bar{E}_o=E_o \otimes I_3 \in \mathbb{R} ^{3n \times 3m} $, $I_3$ is the 3-dimensional identity matrix and $ \otimes$ is the Kronecker product. The orthogonal projection operator $ P : \mathbb{R}^3 \to \mathbb{R}^{3\times 3}$ is defined as
\begin{align}
	P(v)=I_3-\frac{v}{\|v\|}\frac{v^T}{\|v\|}.   \label{eq21}
\end{align}
This operator projects any non-zero vector $v \in \mathbb{R}^3$ to its orthogonal complement. $\nabla_xV$ denotes the partial differentiation of $V$ with respect to $x$. For each agent $v_i \in \mathcal{V}$, the state $ R_i \in SO(3)$ is devoted such that the matrix group $ SO(3)=\{R \in \mathbb{R}^{3 \times 3} : R^TR=I_{3 \times 3}, det(R)=1\}$ is the space of 3D rotations, where $I_{3 \times 3}$ indicates the three dimensional identity matrix. The group identity indicated by $e$, equals the identity matrix $I_{3 \times3 }$. The tangent space of $SO(3)$ in the identity element of the group, is the Lie algebra of $SO(3)$ written as $\mathfrak{so}(3)=T_e SO(3)= \{ \Omega \in \mathbb{R}^ {3 \times 3} : \Omega^T = -\Omega \} $, which is the space of all $3 \times 3$  skew symmetric matrices. The tangent space of $SO(3)$ at any group member $R$ is $ T_R SO(3)=\{RV: V \in \mathfrak{so}(3)\}$. 
There is a map between the tangent vector in Lie group $ \Omega=\hat{\omega} \in \mathfrak{so}(3)$ with a vector in $\omega \in \mathbb{R}^3$ using the following $hat (\hat{.}) $ and $vee (\check{.}) $  maps 
\begin{align}
\omega=
\left(
\begin{array}{c}
\omega_1 \cr\omega_2\cr\omega_3
\end{array}
\right)
\in \mathbb{R}^3
\stackrel[(\check{.})]{(\hat{.})}{\rightleftharpoons}
\hat{\omega}=
\left(
\begin{array}{ccc}
0&-\omega_3&\omega_2\cr\omega_3&0&-\omega_1\cr-\omega_2&\omega_1&0
\end{array}
\right)
\in \mathfrak{so}(3) \label{eq24}
\end{align}

\subsection{Problem Formulation}
Consider a team of $n$ agents in $SE(3)$. The kinematics of the agent $i$, $i\in\{1\cdots,n\}$ is expressed as
\begin{align}
\left\{ 
\begin{array}{ccc}
\dot{R}_i &= &R_i \hat{\omega}_i    \label{eq1} \\
\dot{p}_i &= &R_i v_i
\end{array}  \right.
\end{align}
where $p_i=[p_{xi}, p_{yi}, p_{zi}]^T \in \mathbb{R}^3$ is the position of the $i$-th agent and $R_i \in SO(3)$ is the rotation matrix associated with the orientation of it. Furthermore, $v_i , \omega_i \in \mathbb{R}^3$ are respectively the linear and angular velocity of agent $i$ with respect to the body frame. 
For notation convenience, we use the position and attitude of the complete configuration in the stacked form as $\chi(i)=(\chi_p(i),\chi_r(i))=(p_i,R_i) \in SE(3)$.

The relative bearing measurement between agents $i$ and $j$ associated to the edge $e_k=(v_i,v_j)$ is defined as  \\
\begin{align}
b_k=b_{ij}=R_i^T \frac{p_i-p_j}{\|p_i-p_j\|}=R_i^T\bar{p}_{ij}	\label{eq22}
\end{align}
and
\begin{align*}
\bar{p}_{ij}=d_{ij} p_{ij}, \ \ p_{ij}=p_i-p_j, \ \ d_{ij}=\frac{1}{\|p_i-p_j\|}.
\end{align*}
\textbf{Definition 1}. The bearing rigidity function for a framework  $(\mathcal{G}, \chi )$ of n-agent is expressed as the map
\begin{align}
	b_{\mathcal{G}}: SE(3)^n \to  \mathbb{S}^{2m}   \nonumber \\
	b_{\mathcal{G}}(\chi) = [ b_1^T \dots b_m^T]^T  \label{eq23} 
\end{align}
that can be written in the following compact form
\begin{align}
	b_\mathcal{G}(\chi)=\text{diag}(d_{ij} R_i^T)\bar{E}^Tp .  \label{eq7}
\end{align}
The bearing rigidity matrix in $SE(3)$ is defined as the gradient of the rigidity function as $ \text{B}_\mathcal{G}(\chi)= \nabla_{\chi} b_\mathcal{G}(\chi) \in \mathbb{R}^{3m \times 6n} $ that can be expressed as
\begin{align}
	\text{B}_\mathcal{G}(\chi) &= \left[ \frac{\partial b_{ij}}{\partial p_i} \quad \frac{ \partial b_{ij}}{\partial R_i} \right]  \nonumber\\
	& =\left[- \text{diag}(d_{ij} R_i^TP(\bar{p}_{ij})) \bar{E}^T \quad -\text{diag}(R_i^T\widehat{\bar{p}_{ij}}) \bar{E}_o^T \right]   \label{eq8}
\end{align}
such that $\widehat{\bar{p}_{ij}}$ is defined as (\ref{eq24}).

\section{Bearing Formation Control in SE(3) with bearing-only constraints}
In this section, the formation control problem in $SE(3)$ with no common reference frame is studied. The desired formation is proposed as bearing only constraints. A network of n-agent is considered with dynamics presented in (\ref{eq1}). The control strategy to transmit the agents to the desired formation is based on the gradient field of a potential function $ \phi$. 
This function is defined to describe the bearing and/or distance constraints.
For converging to the desired formation that is addressed in this section, the potential function is only introduced in terms of the bearing only constraints as
\begin{align}
\phi(\chi)=\frac{1}{2} \| b_\mathcal{G}-b_\mathcal{G}^* \|^2                        \label{eq2}
\end{align}
while $b_\mathcal{G}^*$ is the desired bearing rigidity function. The control input is considered as the negated gradient of this potential function as
\begin{align}
\binom{v_i}{\omega_i}=-k \nabla_\chi \phi   \label{eq3}
\end{align}
while
\begin{align*}
\nabla_\chi \phi=\binom{\nabla _{p_i} \phi}{\nabla _{R_i} \phi}.
\end{align*}
Whereas
\begin{align*}
	\nabla_\chi \phi &= \nabla_\chi b_\mathcal{G}^T \frac{\partial \phi}{\partial b_\mathcal{G}}\\
	&= \text{B}_\mathcal{G}^T(\chi) \frac{b_\mathcal{G}-b_\mathcal{G}^*}{\|b_\mathcal{G}-b_\mathcal{G}^*\|} =-\text{B}_\mathcal{G}^T(\chi) \frac{b_\mathcal{G}^*}{\|b_\mathcal{G}-b_\mathcal{G}^*\|} \\
	&=\left[ \text{diag}(d_{ij} R_i^TP(\bar{p}_{ij})) \bar{E}^T \quad \text{diag}(R_i^T\widehat{\bar{p}_{ij}}) \bar{E}_o^T \right]^T\\ 
	& \ \ \ \ \ \frac{b_\mathcal{G}^*}{\|b_\mathcal{G}-b_\mathcal{G}^*\|}.
\end{align*}
As a result, the local body frame control signal is expressed as
\begin{align}
	\left\{ 
	\begin{array}{ccc}
		v_i &=\sum_{(i,j) \in \mathcal{E}}^{} -d_{ij} R_i^T P(\bar{p}_{ij}) b_{ij}^*  \\
		\omega_i &=\sum_{(i,j)\in \mathcal{E}} -R_i^T \widehat{\bar{p}_{ij}} b_{ij}^*   \quad \ \label{eq4}   \\ 
	\end{array}   \right.
\end{align}
In the same way as \cite{zhao2016bearingalmost}, the following theorem is used for guaranteeing the stability of the formation.
\newtheorem{lem}{Lemma}
\begin{lem}  \label{lem1}
	The control law (\ref{eq4}) keeps the centroid $\bar{p}$ and the scale $s$ invariant 
\end{lem}
\begin{proof}
	Due to lack of space the proof is omitted. Because of the same procedure, for more information, please refer to \cite{zhao2016bearingalmost}.
\end{proof}
\begin{figure*}[!ht]
	\centering
	\subfloat[Behavior of agents from initial value to the final formation]{\includegraphics[width=0.33\textwidth]{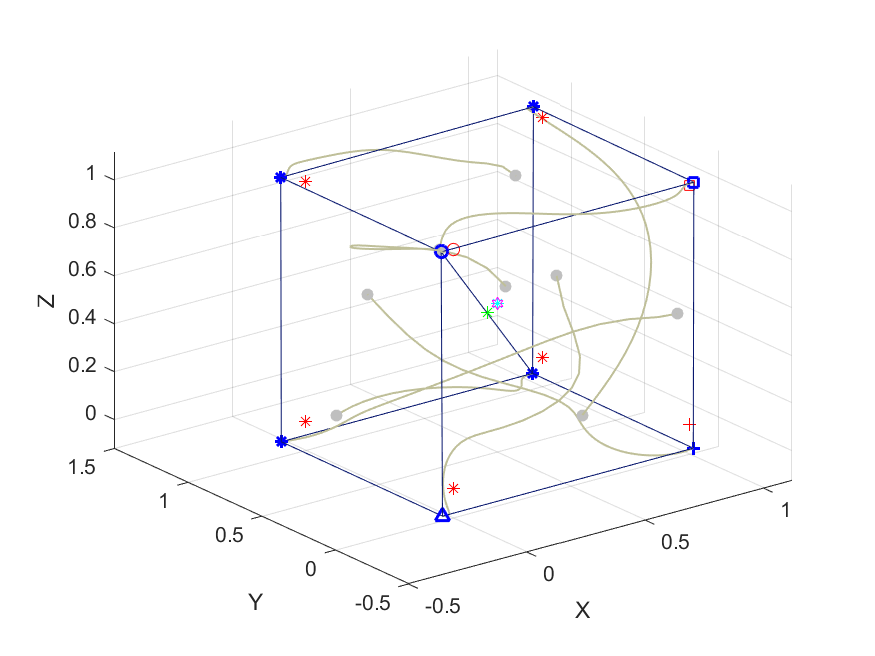}}
	\qquad 
	\qquad
	\subfloat[Error signal norm between agents bearing and the desired bearing]{\includegraphics[width=0.33\textwidth]{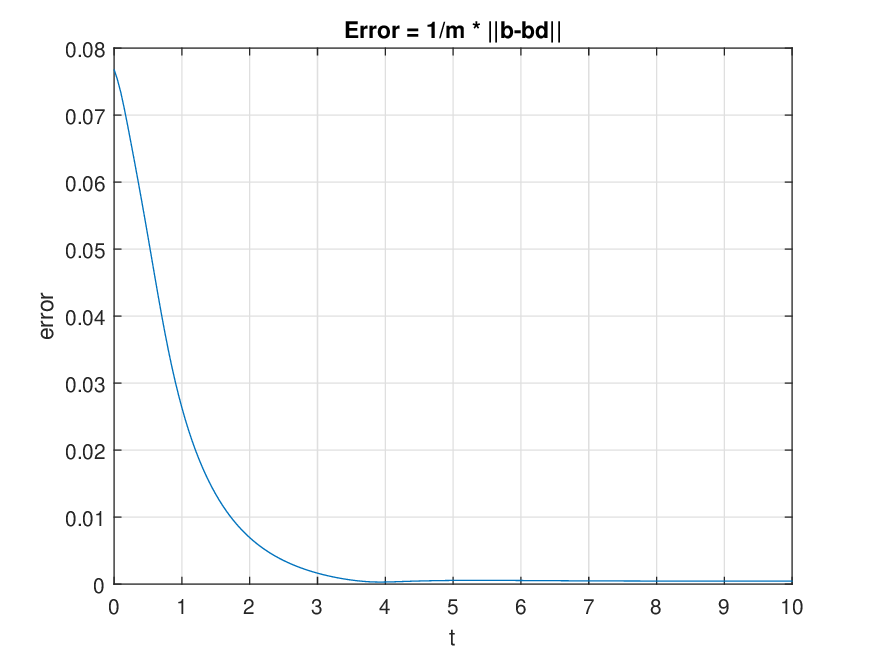}}
	\caption{Bearing-only formation and error for a network of 8 agents}
	\label{fig:1}
\end{figure*}
\section{Bearing-distance based Formation Control in $SE(3)$}
\subsection{Bearing-Distance Rigidity Theory in $SE(3)$}
In this section, the desired formation is considered in terms of both distance and bearing constraints. Besides, the control strategy requires mixed distance and bearing information. \\
Again, consider a network of $n$ rigid bodies as a multi-agent system modeled as a framework $(\mathcal{G} , \chi)$ with kinematics defined as (\ref{eq1}). We have the following definitions. \\
\textbf{Definition 2:}
Suppose the edge set $ \mathcal{E}$ in the graph $ \mathcal{G}( \mathcal{V}, \mathcal{E}) $ consists of two edge sets of $ \mathcal{E}_B$ and $ \mathcal{E}_D$ which are bearing and distance measurements respectively, such that $ \mathcal{E}_B \cup \mathcal{E}_D = \mathcal{E}$, while none of $ \mathcal{E}_D$ and $\mathcal{E}_B$ are not zero. So the network graph is defined by $ \mathcal{G}( \mathcal{V}, \mathcal{E}_B, \mathcal{E}_D) $. In the case that $ (i,j) \in \mathcal{E}_B $ and $(i,j) \in \mathcal{E}_D $, it means there exists two edges between agents $i$ and $j$, namely, both bearing and distance measurements are available between these agents. Besides, $m_b$ indicates the number of available bearing measurements between agents, and $m_d$ is the number of available distance measurements between agents. \\
\textbf{Definition 3:} Rigidity Function equals the map $F_{\mathcal{G}}$ defined as
\vspace{-0.3cm}
\begin{align*}
	F_\mathcal{G}: \mathbb{R}^{3n} & \to \mathbb{R}^{(m_d+m_b)} \\
	\chi & \mapsto F_\mathcal{G} (\chi) 
\end{align*}
while \\
$ \chi=(p,R)=\{ (p_1,R_1), \dots (p_n,R_n) \} \in (SE(3))^n$\\
$p=[ p_1^T \dots p_n^T ]^T \in \mathbb{R}^{3n} $\\
$R=[ \ R_1^T \dots R_n^T ]^T \in (SO(3))^n$.\\
Inspiring by \cite{bishop2013stabilization}, the vector $F_\mathcal{G}$ dedicates all the available bearing and distance measurements on stack form as
\vspace{-0.2cm}
\begin{align}
F_\mathcal{G}(\chi)= 
\left[
\begin{array}{c}
F_1\\
\vdots\\
F_{m_d}\\ \hdashline[2pt/2pt]
F_{m_{d+1}}\\
\vdots\\
F_{m_d+m_b}
\end{array}
\right]
 \triangleq \label{eq10}
\left[
\begin{array}{c}
\tfrac{1}{2} z_1^2\\
\vdots \\
\tfrac{1}{2} z_{m_d}^2\\ \hdashline[2pt/2pt]
b_1 \\
\vdots\\
b_{m_b}
\end{array}
\right]
= \left[ 
\begin{array}{c}
\textcolor{white}{z}\\
d_\mathcal{G} \\
\textcolor{white}{z}\\ \hdashline[2pt/2pt]
\\
b_\mathcal{G}\\
\\
\end{array}
\right]
\end{align}
with the following parameter definitions
\begin{align*}
	b_k & =b_{ij}=R_i^T \frac{p_i-p_j}{\|p_i-p_j \|}=R_i^T \bar{p}_{ij}   \\
	e_k & =p_i-p_j\\
	z_k & =z_{ij}= \|e_k\|=\|p_i-p_j\| , k=1,...,m_d\\
	d_{\mathcal{G}} & = [\tfrac{1}{2}z_1^2 \dots \tfrac{1}{2}z_{m_d}^2]^T \in 	\mathbb{R}^{m_d}\\
	b_{\mathcal{G}} & =[b_1 \dots b_{m_b}]^T \in \mathbb{R}^{3m_b}.
\end{align*}
To clarify, all parameters with $\mathcal{G}$ indices indicate that parmeter in stack form which related to the whole graph. \\
As a result, the $SE(3)$-Bearing Rigidity Matrix $R_\mathcal{G}(\chi) \in \mathbb{R}^{(m_d+3m_b) \times 6n} $ for the framework $ (\mathcal{G}, \chi)$ is defined as follows
\begin{align}
	&R_\mathcal{G}(\chi)  =\nabla_{\chi} F_{\mathcal{G}(\chi)}  \nonumber \\
	&=\left[ \nonumber
	\begin{array}{ccc;{2pt/2pt}ccc} 
	\frac{\partial F_1}{\partial p_1} & \ldots & \frac{\partial F_1}{\partial p_n} & \frac{\partial F_1}{\partial R_1} & \ldots & \frac{\partial F_1}{\partial R_1} \\ 
	\vdots & \ddots & \vdots & \vdots & \ddots & \vdots\\
	\frac{\partial F_{m_d}}{\partial p_1} & \ldots & \frac{\partial F_{m_d}}{\partial p_n} & \frac{\partial F_{m_d}}{\partial R_1} & \ldots & \frac{\partial F_{m_d}}{\partial R_n}\\ \hdashline[2pt/2pt]
   \frac{\partial F_{m_{d+1}}}{\partial p_1} & \ldots & \frac{\partial F_{m_{d+1}}}{\partial p_n} & \frac{\partial F_{m_{d+1}}}{\partial R_1} & \ldots & \frac{\partial F_{m_{d+1}}}{\partial R_n}\\
	\vdots & \ddots & \vdots & \vdots & \ddots & \vdots\\
	\frac{\partial F_{m_d+m_b}}{\partial p_1} & \ldots & \frac{\partial F_{m_d+m_b}}{\partial p_n} & \frac{\partial F_{m_d+m_b}}{\partial R_1} & \ldots & \frac{\partial F_{m_d+m_b}}{\partial R_n}\\
	\end{array}
	\right]\\
	\nonumber \\
	&=\left[  \label{eq11}
	\begin{array}{c;{2pt/2pt}c}
	D_\mathcal{G} & 0 \\ \hdashline[2pt/2pt]
	G_\mathcal{G} & K_\mathcal{G}
	\end{array}
	\right].
\end{align}
\begin{figure*}[!h]
	\begin{minipage}[c]{0.48\textwidth}
		\centering
		\includegraphics[width=0.75\textwidth]{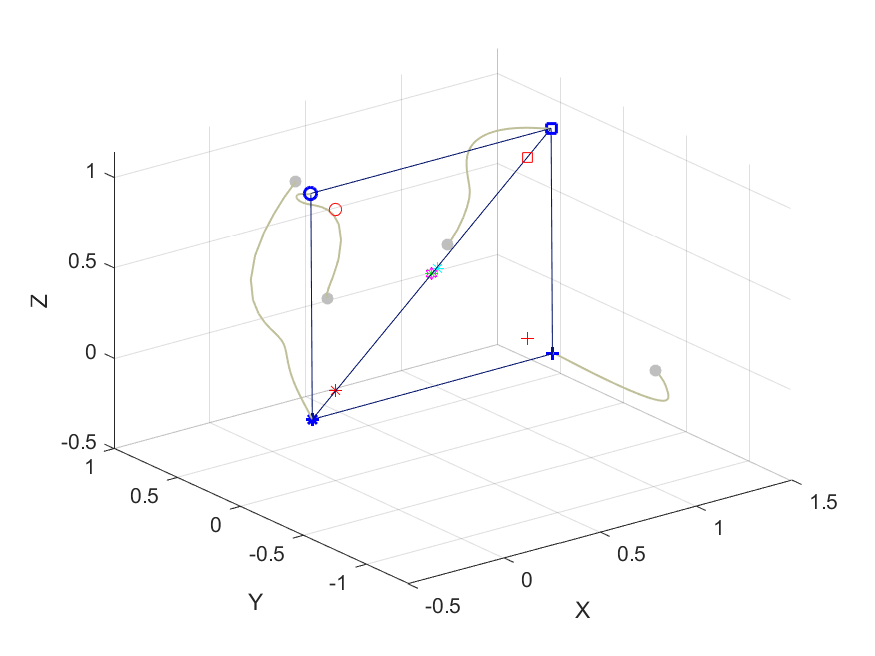}
		\caption[short text]{Behavior of agents from initial value to the final\\ formation \qquad \qquad}
		\centering
		\label{fig:2}
	\end{minipage}
	\begin{minipage}[c]{0.48\textwidth}
		\centering
		\includegraphics[width=0.75\textwidth]{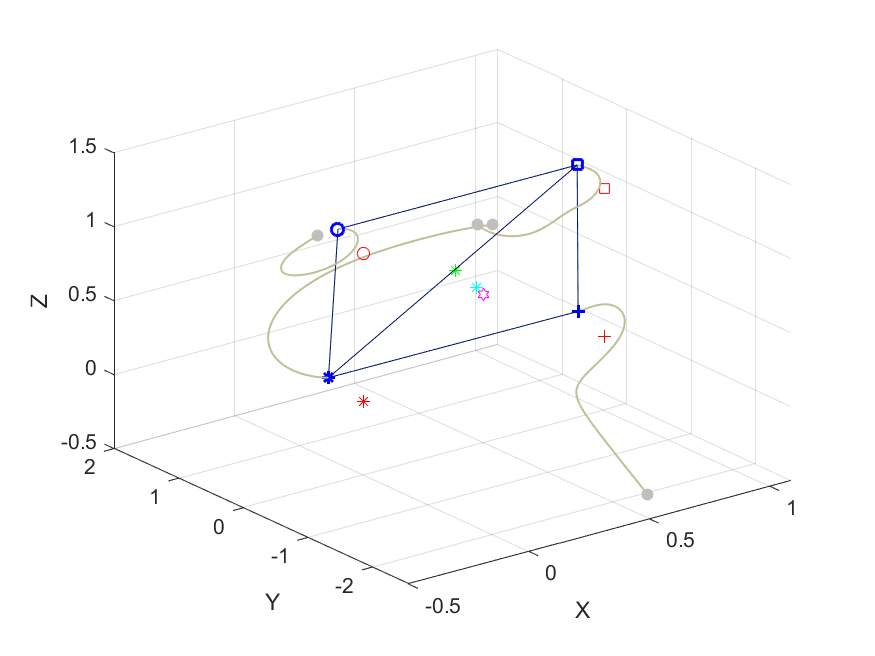}
		\caption[short text]{Formation for the case with 3-bearing and 4-distance constraints }
		\centering
		\label{fig:3}
	\end{minipage}
\end{figure*}
It can be observed that the rigidity matrix (\ref{eq11}) equals the Jacobian of the $SE(3)$-rigidity function. It has been indicated that the null-space of the $SE(3)$-rigidity matrix $R_\mathcal{G}(\chi)$ specifies the set of infinitesimal motions of the framework \cite{michieletto2016bearing}.
The above rigidity matrix is interpreted as the bearing part as $B_\mathcal{G}=[G_\mathcal{G} \  K_\mathcal{G}]$ and distance part $D_\mathcal{G}$. Calculating all parameters of the rigidity matrix $R_\mathcal{G}$ yields
\begin{align}
	D_\mathcal{G} & =J(e) \bar{E};  \label{eq12}\\
	G_\mathcal{G} & =-diag(d_{ij}R_i^TP(\bar{p}_{ij})) \bar{E}^T \label{eq13}\\
	K_\mathcal{G} & =-diag(R_i^T \hat{\bar{p}}_{ij}) \bar{E}_o^T  \label{eq14}
\end{align} 
while
\begin{align*}
	J(e)=diag(e_1^T, ..., e_{m_d}^T)
\end{align*}
and
\begin{align}
	R_\mathcal{G}= 
	\left[
	\begin{array}{c;{2pt/2pt}c}
	\\
	J(e) \bar{E} & 0\\
	\\ \hdashline[2pt/2pt]
	\\
	-diag(d_{ij}R_i^TP(\bar{p}_{ij})) \bar{E}^T & -diag(R_i^T \hat{\bar{p}}_{ij}) \bar{E}_o^T 
	\end{array}
	\right].  \label{eq15}
\end{align}
\subsection{Formation Control in $SE(3)$}
The formation control problem in SE(3) is formulated in this section to steer a network of $n$ rigid bodies with kinematics equation as (\ref{eq1}) to a desired formation consists of bearing and distance constraints.
The cost function is introduced based on both bearing and distance constraints as follows
\begin{align}
\phi(\chi)=\frac{1}{2} \| b_\mathcal{G}-b_\mathcal{G}^* \|^2+ \frac{1}{2}\|d_\mathcal{G}-d_\mathcal{G}^*\|^2
\end{align}
while $b_\mathcal{G}^*$ is the desired bearing rigidity function and $d_\mathcal{G}^*$  is the desired distance function. The control input is similarly expressed as (\ref{eq3}). As a result, the gradient of the cost function $ \phi$ can be computed using the chain rule as 
\vspace{-0.3cm}
\begin{align}
	\left[
	\begin{array}{c}
	v_i\\
	\omega_i
	\end{array}
	\right]= 
	-\nabla_\chi \phi= 
	\left[
	\begin{array}{c}
	-\nabla_{p_i} \phi\\
	-\nabla_{R_i} \phi
	\end{array}
	\right]   \label{eq17}
\end{align}
such that
\begin{align}
	\nabla_{p_i} \phi & = \nabla_{p_i} b_\mathcal{G} (b_\mathcal{G}-b_\mathcal{G}^*)+ \nabla_{p_i} d_\mathcal{G} (d_\mathcal{G}-d_\mathcal{G}^*)  \nonumber\\
	& = G_\mathcal{G} ^T (b_\mathcal{G}-b_\mathcal{G}^*) + D_\mathcal{G}^T (d_\mathcal{G}-d_\mathcal{G}^*)   \nonumber\\
	& =-(diag(d_{ij}R_i^TP(\bar{p}_{ij})) \bar{E}^T)^T (b_\mathcal{G}-b_\mathcal{G}^*) \nonumber \\ 
	& \ \ \ \ +  \bar{E}^T J(e)^T (d_\mathcal{G} - d_\mathcal{G}^*)  \nonumber \\ 
	& = (diag(d_{ij}R_i^TP(\bar{p}_{ij}))\bar{E}^T)^T b_\mathcal{G}^* +  \bar{E}^T J(e)^T (d_\mathcal{G} - d_\mathcal{G}^*)  \label{eq16} \\
	\nabla_{R_i} \phi & =\nabla_{R_i} b_\mathcal{G} (b_\mathcal{G}-b_\mathcal{G}^*)=K_\mathcal{G}^T(b_\mathcal{G}-_\mathcal{G}^*)  \nonumber \\
	& = -\bar{E}_o (diag(R_i^T \hat{\bar{p}}_{ij}))^T(b_\mathcal{G}-b_\mathcal{G}^*) \nonumber \\
	& = \bar{E}_o (diag(R_i^T \hat{\bar{p}}_{ij}))^T  b_\mathcal{G}^*. \label{eq18}
\end{align}
It should be noted that $ P(\bar{p}_{ij})R_i b_g=0$. Therefore, the resulting control law navigate the agent to the desired formation is given by (\ref{eq17}) as
\begin{align}
	\left[
	\begin{array}{c}
	v_i\\
	\omega_i
	\end{array}
	\right]= 
	\left[
		\begin{array}{c}
		\!\!\!\!\!\! -(diag(d_{ij}R_i^TP(\bar{p}_{ij}))\bar{E}^T)^T b_\mathcal{G}^* \! +  \bar{E}^T J(e)^T (d_\mathcal{G} - d_\mathcal{G}^*) \!\!\!\!\!\! \\ 
		\vspace*{-0.25cm}\\
		-\bar{E}_o (diag(R_i^T \hat{\bar{p}}_{ij}))^T  b_\mathcal{G}^*
		\end{array}
	\right].   \label{eq20}
\end{align}
\newtheorem{theorem}{Theorem}
\begin{theorem}
	Using the control input (\ref{eq20}), the centroid $\bar{p}$ and the scale $s$ are invariant. while
	\begin{align}
	\bar{p}=1/n \sum_{1=1}^{n}p_i\ \ \  , \ \ \ s=\sqrt{\frac{1}{n} \sum_{i=1}^{n} \|p_i-\bar{p} \|^2}  \label{eq19}
	\end{align}
\end{theorem}
This theorem represents the target formation and the initial formation have the same scale and centroid.
\begin{proof}
	Centroid: It should be concluded that  $\bar{p}^*=\bar{p}(0)$.
	\begin{align*}
		\dot{\bar{p}}= &\frac{1}{n} (1_n^T \otimes I_3) diag(d_{ij} R_i^T P(\bar{p}_{ij})) \bar{E}^T b_g^* \\
		& + \frac{1}{n} (1_n^T \otimes I_3) J(e) \bar{E}(d_\mathcal{G}-d_\mathcal{G}^*).
	\end{align*}
	For the first term, as it is indicated in \cite{michieletto2021unified} that $(1_n^T \otimes I_3) \in \mathcal{S}_t=ker (G_g)$, so $(1_n^T \otimes I_3)  (R_i^T P(\bar{p}_{ij}))^T b_g^* = 0$.
	For the second term, in distance rigidity problem with rigidity matrix $D_\mathcal{G}$  using the fact that $D_\mathcal{G} (I_n \otimes I_3)=0$ results in $(I_n \otimes I_3)J(e) \bar{E}=0$.
	So, $\dot{\bar{p}}=0$ yields $ \bar{p}^*=\bar{p}(0)$. \\
	Scaling: Since (\ref{eq19}) can be written as the centroid term such that $ s=\| p-(1_n \otimes \bar{p}) \|/\sqrt{n}$, so
	\begin{align*}
		\dot{s}= \frac{1}{\sqrt{n}} \frac{ (p-(1_n\otimes\bar{p}))^T}{\| p-(1_n\otimes\bar{p}) \|}\dot{p}.
	\end{align*}
	It is obtained from $\dot{p} \perp p$ that $p^T \dot{p}=0$, and from the previous part $\dot{p} \perp (1_n \otimes \bar{p})$, thus $\dot{s}=0$ and $s^*=s(0)$.
\end{proof}
\begin{figure*}[!h]
	\begin{minipage}[c]{0.48\textwidth}
		\vspace*{0pt}
		\centering
		\includegraphics[width=0.75\textwidth]{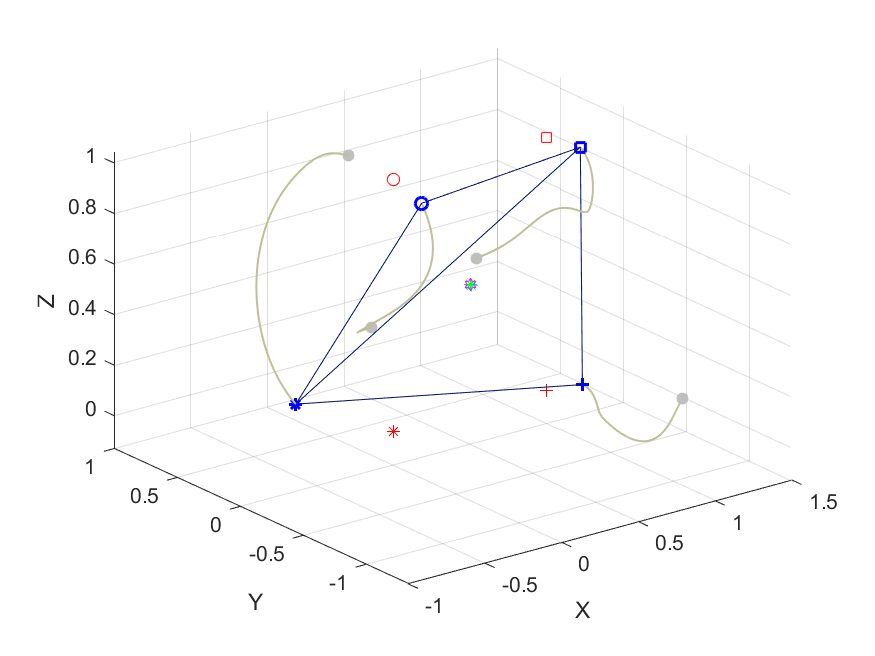}
		\centering
		\caption{Formation for the case 3-bearing and 3-distance\\ constraints}
		\label{fig:4}
	\end{minipage}
	\begin{minipage}[c]{0.48\textwidth}
		\vspace*{0pt}
		\centering
		\includegraphics[width=0.75\textwidth]{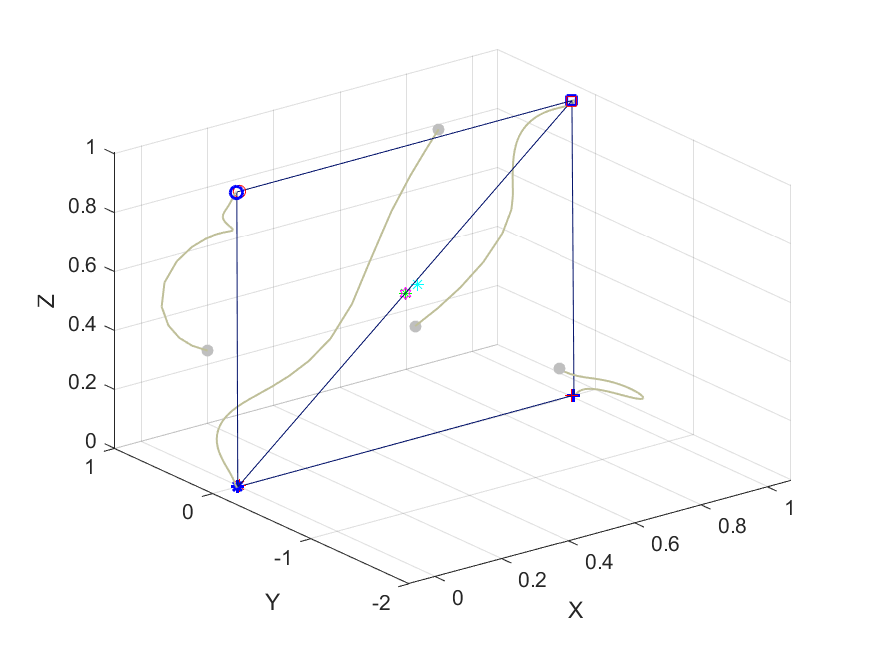}
		\centering
		\caption{Formation for the case 5-bearing and 1-distance constraints \\ }
		\label{fig:5}
	\end{minipage}
\end{figure*}
\section{simulation results}
Numerical simulations of the proposed method is presented in this section to verify the validity of the investigated control law algorithm. In the first simulation, a network of 8 agents evolves under the bearing-based control law (\ref{eq4}). The next example illustrates the mixed bearing and distance final constraints formation.\\
\textit{Example 1.} Eight agents  with bearing edges are considered as a network under (\ref{eq1}) dynamical equations that should be converged to a cubical desired formation which is defined based on bearing-only constraints, namely with the cost function considered as (\ref{eq2}). Applying the control law (\ref{eq4}) to this system results in a fast convergence to the desired formation. Fig. (\ref{fig:1}) shows the trajectories agents travel from initial to final orientation. Gray dots are the indicator of the initial configuration of agents and  gray lines are the paths they pass until reach the final formation that is marked with blue signs. Red signs depict the desired formation. It is obvious from the figure that agents reach the target bearing formation but in different scaling. It means bearing constraints are completely satisfied that is confirmed by Fig. (\ref{fig:1})-b showing the bearing error between the desired and final formation.

\textit{Example 2.} Four agents are evolved under the control law (\ref{eq4}). As it is conceived from Fig. (\ref{fig:2}), the network converges to the target bearing formation; however, as the control law is only based on bearing constraints, the final formation is in a different scaling with the desired formation. So, in this example, mixed bearing and distance constraints implemented to the system which results in an exact final formation with the same scale as the desired one. The control law (\ref{eq20}) is applied to the system of four agents with different desired mixed constraints. First of all, the target formation includes 3 bearing and 4 distance constraints. As it is seen from Fig. (\ref{fig:3}), bearing constraints are not totally satisfied. Fig. (\ref{fig:4}) shows the problem with 3 bearing and 3 distance constraints. The graph has not been succeeded in achieving convergence to the desired formation. Finally, the combination of 5 bearing, 1 distance constraints leads to the exact formation as the desired one. From Fig. (\ref{fig:5}), it is seen that the final formation is the same as the desired one and both scaling and coordinated rotation properties of the formation are strongly satisfied. Due to the lack of space, the minimum quantities of the required bearing and distance constraints under which the desired formation is obtained will be investigated in the next paper.


\section{Conclusion}
In this paper, bearing formation control problem is implemented in $SE(3)$ manifold. The control law is constructed based on cost function and gradient descent method. Then the problem is extended to the mix distance and bearing formation constraints, and conditions under which the problem has solution is considered.   

\end{document}